\documentclass[submission, Phys]{SciPost}

\def\supp{\text{supp}}
\usepackage{tikz}
\tikzset{node distance=2cm, auto}
\usetikzlibrary{arrows}
\usetikzlibrary{calc}

\newcommand{\be}[0]{\begin{equation}}
\newcommand{\ee}[0]{\end{equation}}

\numberwithin{equation}{section}
\newcommand{\beq}{\begin{eqnarray}}
\newcommand{\eeq}{\end{eqnarray}}

\usepackage{amsthm}
\theoremstyle{definition}

\theoremstyle{lemma}
\newtheorem*{lemma}{Lemma}
\theoremstyle{theorem}
\newtheorem*{theorem}{Theorem}
\theoremstyle{proposition}

\newtheorem*{corollary}{Corollary}
\theoremstyle{remark}

\hypersetup{
    colorlinks,
    linkcolor={red!50!black},
    citecolor={blue!50!black},
    urlcolor={blue!80!black}
}
\usepackage[bitstream-charter]{mathdesign}
\DeclareSymbolFont{usualmathcal}{OMS}{cmsy}{m}{n}
\DeclareSymbolFontAlphabet{\mathcal}{usualmathcal}

\begin{document}


\begin{center}{\Large \textbf{
Nonlocal Conformal Field Theory}}\end{center}
\begin{center}
Bora Basa\textsuperscript{1*}, Garbrile La Nave\textsuperscript{2}  and Philip Phillips\textsuperscript{1}
\end{center}

\begin{center}
{\bf 1}
Department of Physics and Institute for Condensed Matter Theory,
University of Illinois
1110 W. Green Street, Urbana, IL 61801, U.S.A.\\
{\bf 2} Department of Mathematics, University of Illinois, Urbana, Il. 61801
\\

${}^*$ {\small \sf basa2@illinois.edu}
\end{center}

\begin{center}
\today
\end{center}


\section*{Abstract}
{\bf
Using the recently developed notion of a fractional Virasoro algebra, we explore the implied operator product expansions in nonlocal conformal field theories and their geometric meaning. We probe the interplay between classical nonlocality in the functional-analytic sense and quantization in a two-dimensional setting and find that nonlocal quantum dynamics realize this fractional Virasoro algebra exclusively with a state dependent central charge. Notably, we prove that the widely studied free Gaussian fixed points with a fractional Laplacian kinetic term does not fit this criterion but that the RG flow associated non-Gaussian fixed points do. 
}

\noindent\rule{\textwidth}{1pt}

\section{Introduction and summary}

While locality in some form is usually a taken to be a key tenet of field theory, there are numerous settings in which nonlocality is either realized explicitly or is emergent. A mathematically precise example of the former comes from a physical interpretation of the Caffarelli-Silvestre (CS) extension theorem~\cite{caffarelli_extension_2007}: Local bulk dynamics can have corresponding boundary dynamics governed by a nonlocal operator, the Laplacian raised to a non-integer power. 


Quite generally, the fractional Laplacian $(-\Delta)^\gamma$ (or its conformal extension, the Panietz operator\cite{graham,gz,chang}) on a function $f$ in ${\mathbb R}^n$ provides a  Dirichlet-to-Neumann map for a function $\phi$ in ${\mathbb R^{n+1}}$ that satisfies a \emph{local} second-order elliptic differential equation. Paulos, et al.~\cite{paulos_conformal_2016}, simultaneous with ours~\cite{gl2}, noted that the extension theorem allows one to equate the action of a free massive theory in $d+1$ dimensions in a spacetime such as AdS with a $d-$dimensional nonlocal one given by the first term in Eq.~\eqref{nonlocalS}. That is, the CS extension theorem provides another way of understanding the precise way in which conformal invariance arises in the AdS/CFT conjecture~\cite{gubser1,Witten1998}. 

In the realm of statistical physics, which will furnish a concrete example for the framework of nonlocal CFT we will develop, the long-range Ising (LRI) model,
\begin{equation}
H=-J\sum_{i,j} s_is_j/r_{ij}^{d+\gamma},
\end{equation}
where the sum is pairwise over all site,  has nonlocal kinetic energy operators appearing in the corresponding field theory
\begin{equation}\label{nonlocalS}
S=\int d^dx \left(\frac12\Phi(x)(-\Delta)^{\gamma}\Phi(x)+\frac{g}{4!}\Phi(x)^4\right).
\end{equation}
 This form of the action is immediate, given
\begin{equation}\label{pseudo}
(-\Delta)^\gamma f(x)=C_{d,\gamma}\int d^dx'\frac{f(x)-f(x')}{|x-x'|^{d+2\gamma}}.
\end{equation}
The IR physics associated to this continuum model has received a great deal of attention in past decades. As we will review in Sec.~\ref{sec:LRI}, the fractional/nonlocal $\phi^4$ theory flows to either the trivial fractional Gaussian fixed point or to a fixed point in either the long-range or short-range Ising (SRI) universality class~\cite{fisher_critical_1972,sak_recursion_1973,sak_1977}. Particularly relevant to this work is the conformal invariance of these fixed points. While the fractional Gaussian~\cite{rajabpour_conformal_2011} and SRI fixed points are known to be conformally invariant, the LRI fixed point has recently been demonstrated to also possess conformal symmetry~\cite{paulos_conformal_2016}. 


Irrespective of conformal invariance, nonlocal field theories have rather peculiar properties. We have shown previously that, in some cases~\cite{basa_classification_2020}, the classical `nonlocality' of the free fractional Gaussian theory (Eq.~\eqref{pseudo} with g=0) does not survive quantization because it can be removed by a suitable \emph{field redefinition} at the level of the path integral. This surprising interaction between quantization and nonlocality naturally suggests a deeper examination of the symmetries such theories possess and the anomalies associated to them. While we tackle the problem of conformal invariance in the present work, a more general (an ambitious) exercise is to study generic current algebras through this lens.   

In this paper we realize fractional theories possessing classical conformal invariance, in the 2 dimensional setting, as ones governed either by the tensor product of the ordinary Virasoro algebra and a particular commutative algebra, $\mathcal H$, of holomorphic functions on $\mathbb C$ or, more generally, by a fractional Virasoro algebra~\cite{la_nave_fractional_2019} with the $\mathcal H$-structure now carried by the central charge. By constructing the OPEs of the nonlocal stress tensor with itself and conformal primaries, we find that the OPEs of a general nonlocal CFT are compatible with the underlying $\mathcal H$ structure of the fractional Virasoro algebra. In particular the $TT$ OPE encodes the $\mathcal H$-structure as a form of state dependence - a somewhat bizarre notion that, curiously, has emerged independently in addressing the information paradox~\cite{Raju1,Raju2}.
We then prove that a fractional CFT whose central charge $c$ is constant (or not state-dependent), is equivalent, in an appropriate {\it covering}, to a local CFT which provides a practical probe of locality. We then extend this idea to deformations of such CFTs to state an analog of the C-theorem highlighting the importance of interactions in the emergence of nonlocality. We conclude with a discussion of the LRI-SRI crossover problem from the perspective of this new nonlocal refinement of perturbative RG.


\section{Fractional Virasoro algebra}

In $d=2$, let the \emph{fractional holomorphic pseudo-derivative} be defined as
\begin{equation}
  \partial_z^\gamma f(z)= \frac{\Gamma(\gamma+1)}{2\pi i} \int_\Xi \frac{f(\xi)}{(\xi-z)^{1+\gamma}}d\xi
\end{equation}
where the contour $\Xi$, around $z$ should be understood as being lifted to the universal covering due to the branch cut associated with the non-integer power. 

We think of this object relative to the fractional Laplacian, 
\begin{equation}
  \partial^\gamma \bar \partial^\gamma f(z,\bar z) = |C_{\gamma,2}|^2 (-\Delta)^\gamma f(z,\bar z),
\end{equation}
and as the most natural analog of derivations of the Laurent polynomial ring in the fractional setting, $\mathbb C[[z^{\pm \gamma}]]$.
\subsection{Algebraic construction}
We now briefly review the construction of the fractional Virasoro algebra of Ref.~\cite{la_nave_fractional_2019}. In order to understand a nonlocal CFT, as a first step, one has to understand a stress tensor generated by nonlocal pseudo-differential operators. At the level of OPEs, the fractionalization of the symbol of the differential operators in the stress tensor gives rise to nontrivial monodromy that must somehow be encoded in the central charge. Of course, the standard Virasoro algebra, being generated locally, is not able to accommodate this structure. We therefore seek the minimal algebraic construction that does.


The standard Witt algebra is the algebra of derivations on the Laurent polynomial ring, $\mathbb{C}[[z^{\pm 1}]]$. We instead consider  fractional pseudo-derivations on $\mathbb{C}[[z^{\pm \gamma}]]$ with $\gamma\in (0,1)$. To better understand these objects, we imagine a graded complex vector space, 
\begin{equation}
    V = \bigoplus_{k\in \mathbb{Z}} V^{\gamma k},
\end{equation}
with each subspace, $V^{\gamma k}$, spanned by $z^{\gamma k}$. In order to define a derivation, we need a degree decrementing linear map
$$
P_k: V^{\gamma k}\to V^{\gamma(k-1)}
$$
defined as $P(z^{\gamma k })= z^{\gamma(k-1)}$. Then, the fractional derivative is written
\begin{equation}\label{eq:deriv}
  \partial^\gamma _z = \bigoplus_k \lambda_\gamma\frac{\Gamma(\gamma k +1)}{\Gamma(\gamma(k-1)+1)}P_k.
\end{equation}
The important generalization here is that the coefficients now are meromorphic functions of $k$ rather than being $k$ itself. The term that depends only on $\gamma$, $\lambda_\gamma$, is arbitrary so long as $\lambda_{\gamma\to 1}=1$. Similarly, the meromorphic coeffcients reduce to $k$ under $\gamma\to 1$. The definitions
$$\Gamma_p^\gamma(s):=\frac{\Gamma(1+(s+p)\gamma)}{\Gamma(1+(s-1+p)\gamma)},\quad A_{p,q}(s):=\Gamma_p^\gamma(s)-\Gamma^\gamma_q(s) $$
will be helpful moving forward. 

We may now ask if $L_n^\gamma :=-z^{(n+1)\gamma}\partial_z^\gamma$ acting on $\mathbb C[[z^{\pm \gamma}]]$ generates some Lie algebra analogous to the Witt algebra. 
A constructive approach one could take here is to compute the bracket $[\phi\otimes L_m^\gamma,\psi\otimes L^\gamma_n]$ with $\phi,\psi$ belonging to some subalgebra of the full algebra of meromorphic functions. These functions are introduced to keep track of the degree dependence of the action of the generators on basis elements:
\begin{equation}
\begin{aligned}
   \phi\otimes  L_n^\gamma(z^{\gamma k}) &= \phi(k)L_n^\gamma(z^{\gamma k})\\
   &=\phi(k)\Gamma_0(k)z^{\gamma (k+n)}
\end{aligned}
 \end{equation} 
 Now,
 \begin{equation}
    [\phi\otimes L^\gamma_m,\psi \otimes L^\gamma_n]=(\psi(m+s)\phi(s)\Gamma^\gamma_m-\psi(s)\phi(n+s)\Gamma_n^\gamma)\otimes L^\gamma_{m+n}
 \end{equation}
prompts the definition of the \emph{fractional bracket},
\begin{equation}\label{eq:bracket}
  [\phi(s),\psi(s)]_{m,n}:=(\psi(m+s)\phi(s)\Gamma^\gamma_m-\psi(s)\phi(n+s)\Gamma_n^\gamma)(1-\delta_{m,n}).
\end{equation}
We thus seek the smallest subalgebra of meromorphic functions that is closed under this bracket to append to the fractional generators to recover the analogous fractional Witt algebra. Calling this subalgebra $\mathcal H$, we may write
\begin{equation}
    \mathcal{W}_\gamma := \bigoplus_{n\in \mathbb{Z}}\mathcal{H}l_n^\gamma.
\end{equation}

$\mathcal W_\gamma$ is not exactly a classical Lie algebra, however. Let $\mathcal A$ be an algebra equipped with integer parameterized product, $\star_{p,q}$. Then a \emph{$\mathcal A$-Lie algebra} is the graded $\mathcal A$ module, $\mathcal W =\oplus_{n\in \mathbb Z}W_n$ endowed with bracket $[\cdot,\cdot]:\mathcal W\times \mathcal W\to \mathcal W$ such that 
\begin{equation}
  [aL, bK] = a\star_{p,q}b[L,K], \quad a,b\in \mathcal A,\; L\in \mathcal W_p,\;K\in \mathcal W_q.
\end{equation}
Equipping the algebra of meromorphic functions holomorphic around $\text{Re }z\in \mathbb Z$ with 
$$
\phi(s)\star_{p,q}\psi(s):=\frac{\psi(p+s)\phi(s)\Gamma^\gamma_p-\psi(s)\phi(q+s)\Gamma_q^\gamma}{\Gamma_p^\gamma(s)-\Gamma^\gamma_q(s)},
$$
we can define the associated bracket in Eq.~\ref{eq:bracket} to characterize the fractional Witt algebra as a \emph{$\mathcal H$-Lie algebra} spanned by $L^\gamma_n$ with
\begin{equation}
   [L_m^\gamma, L_n^\gamma] f(z)=\sum_k f_k A_{m,n}(k)L_{m+n}^\gamma(z^{\gamma k}), \quad f(z) = \sum_n f_n z^{\gamma n} .
 \end{equation} 
If we take $\gamma\to 1$, $\mathcal H\to\mathbb C$. Of course, a $\mathbb C-$Witt algebra is the Witt algebra itself.

We are now in a position to consider central extensions,
\begin{equation}
  0\to \mathcal H\to \mathcal V_\gamma\to\mathcal W_\gamma \to 0.
\end{equation}
The extensions of the fractional Witt algebra by $\mathcal H$ to the fractional Virasoro algebra, $\mathcal V_\gamma$ are indeed $\mathcal H$ parameterized ($H^2(\mathcal W_\gamma,\mathcal H)\cong \mathcal H$). Thus, to a meromorphic function, $c(s)\in \mathcal H$, we associate a fractional Virasoro algebra. This construction indicates that the price we have to pay to \emph{quantize} a classically nonlocal conformal symmetry is a non-constant central charge. The machinery internal to the construction leads us to interpret this as a degree or \emph{state} dependence of a CFT (defined through the set of its correlators) that realizes $\mathcal V_\gamma$.

A central charge that is no longer a c-number but rather an operator is rather difficult to interpret physically. However, by the CS extension mechanism (or some more general analog), one should keep in mind that these fractional theories can be \emph{extended} to well-behaved theories.  Explicitly, the state-dependent formulation of the nonlocal theory is a representation of a state-independent local theory in one higher dimension. This consistency condition also demands that the space of states of the nonlocal CFT be parameterized by the moduli of the higher dimensional theory. For a standard field theoretic application of this idea to quantization, see~\cite{quant_fractional}.

\subsection{Geometric interpretation}\label{sec:geo}

The nonlocal Virasoro algebra that we have defined so far can be given a geometric interpretation: Informally, we envision attaching a Riemann surface to points in a complex domain and an associated Hilbert space so that this base point parameterizes the state dependence of the (soft) breaking of fractional confromal invariance upon quantization.   


Letting $D:= \mathbb C\setminus \{ z= x+iy:\; x\in \gamma \mathbb Z, \; x<0, \; y=0\}$ (here we assume $\gamma \in (0,1)$ {\it irrational}, for simplicity), we define a holomorphic family $\pi: \mathcal X \to D$ of Riemann surfaces. In order to avoid unnecessary complications, we assume that the fibers $\Sigma:=\pi ^{-1} (s)$ with $s\in D$ are all complex-analytically isomorphic and we will take all fibers to possess the same conformal structure. 
 
 Since we can interpret Eq.~\eqref{eq:deriv} as an operator defined on the Laurent polynomial ring in $\mathbb H=\mathbb{C}[[w^{\pm 1}]]$ with $w=z^\gamma$,
we can consider lifts to the universal covering of $\Sigma$. This induces a new family $\hat \pi: \hat {\mathcal X} \to D$ consisting of the universal covering of the fibers of $\pi$.  



 Fixing a Hilbert space on each fiber, $\hat \Sigma$, with an algebra of operators $\mathcal A$, the space of operators on the whole family $\hat{\mathcal X}$ can be thought of as  $\mathcal O(D)\otimes \mathcal A$, where $\mathcal O(D)$ is the function space of $D$ (or the structure sheaf of $D$). The setup is concluded with the realization that $\mathcal O(D)\equiv \mathcal H$ with the linear action on $\mathbb H$ of the form $\phi \star w^{k} = \phi (k)$ for $\phi\in \mathcal H$.


  \emph{We can now envision a form of nonlocal conformal transformations in this context as fiberwise conformal transformations of the family $\hat{\mathcal X}\to D$}. This is done by considering a family of metrics, $g(s)$, on the fibers $\hat\pi ^{-1} (s)$, varying smoothly in the parameter $s$ and then considering fiberwise conformal transformations of this family of metrics. In fact, we think of the exponents $k$ of the states $w^k$ as points in $D$ and for each such point we consider the fiber $\pi^{-1}(k)$ endowed with the metric $g_k$ and we consider `local' conformal transformation (i.e. holomorphic or antiholomorphic maps) of this fiber. On the whole fiber bundle, then, the fiber-wise conformal transformations appear as if twisted nonlocally (in general) because we modify the Witt generators at each level (identifying the level with a point of the base $D$). Via identifying all the fibers, if the coefficients $\phi\in \mathcal H$ of the operators are to be taken constant, this would recover the classical Virasoro algebra.

 Stated in more phydicsal terms, operators, in this setting, have to be defined with respect to their actions on states in general. This is naturally also true for their commutators - hence the deep correspondence between state-dependence and nonlocality.

\section{Stress tensor OPEs}

We work only in the holomorphic sector for the duration of the paper. The \emph{fractional conformal fields} acting on a Hilbert space $\mathbb H$ of states are fractional Laurent series, 
\begin{equation}
  \Phi(z) = \sum_{k\in \mathbb Z}\Phi_k z^{-\gamma k}\in\mathcal H\otimes\text{End}\mathbb H[[z^{\pm\gamma}]],
\end{equation}
with finitely many negative modes. The fractional modes are required to have a prescribed $\mathcal H$-linear action on the Hilbert space. Among these fields, the stress tensor, 
\begin{equation}
  T(z) = \sum_{k\in \mathbb Z}L^\gamma_k z^{-\gamma(k+2)}
\end{equation}
is of special importance as it encodes the \emph{localized} (but not local) fractional conformal invariance. The fractional Virasoro modes are
\begin{equation}
  L^\gamma_n = - z^{\gamma(n+1)}\partial^\gamma = \frac{1}{2\pi i}\oint dz z^{\gamma (n+1)}T(z).
\end{equation}
We remind the reader that the contour integral should be understood as being paired with appropriate lifts. We define a fractional conformal field $\Phi(z,\bar z)$ to be a primary if its variation under a fractional conformal transformation  is
\begin{equation}
   \delta_{\epsilon,\bar \epsilon}\Phi(z,\bar z) =  \left(h_\gamma \partial^\gamma \epsilon +\epsilon \partial^\gamma +\text{conj.} \right)\Phi(z,\bar z).
 \end{equation} 
This is to say, the space of primary operators on $\mathcal X$ has trivial (constant) $\mathcal H$ structure. This ensures compatibility between a well defined Hamiltonian, $H:=L_0^\gamma+\bar L_0^\gamma$, and primary of dimension $h_\gamma$ as we shall demonstrate. The $T\Phi$ OPE for which this variation is recovered is obtained through 
\begin{align*}
        Q_{\epsilon} &:= \oint \frac{dz}{2\pi i} \epsilon(z)T(z)\\
        [Q, \Phi(w)] &=\oint_0 \frac{dw}{2\pi i}\oint_w \frac{dz}{2\pi i}\epsilon(z)  T(z) \Phi(w,\bar w) \\
        &= \oint_0 \frac{dw}{2\pi i}\oint_w \frac{dz}{2\pi i}\epsilon(z)\underbrace{\left(\frac{h C_\gamma }{(z-w)^{1+\gamma}}+\frac{\partial^\gamma_w }{z-w}\right)}_{T(z)\Phi(w)}\Phi(w)\\
        &=\delta_\epsilon \Phi(w)
    \end{align*}
    Using this OPE, one can compute also
    \begin{equation}
      \begin{aligned}
        \left[L^\gamma_n,  \Phi(w)\right] &= \frac{1}{2\pi i}\oint_w dzz^{\gamma (n+1)}T(z) \Phi(w)\\
    &=\frac{1}{2\pi i}\oint_w dzz^{\gamma (n+1)}\left(\frac{h_\gamma }{(z-w)^{1+\gamma}}+\frac{\partial^\gamma_w }{z-w}\right)\Phi(w)\\
    &=h_\gamma \partial_z^\gamma (z^{\gamma (n+1)})\vert_{w}\Phi(w)+ w^{\gamma(n+1)}\partial^\gamma \Phi(w),\\
    &=h_\gamma \Gamma_1^\gamma(n)w^{\gamma n}\Phi(w)+w^{\gamma(n+1)}\partial^\gamma\Phi(w)
      \end{aligned}
    \end{equation}
    for $\gamma n\geq -1$. Setting $n=0$, we see that $L_0^\gamma |h_\gamma\rangle = h_\gamma|h_\gamma\rangle$ makes sense for $|h_\gamma\rangle:=\phi(0)|0\rangle$. 

    We also have the commutator 
 \begin{equation}\label{eq:hamcom}
  \begin{aligned}
  \left[L_0^\gamma, L_{-m}^\gamma \right]=-(\Gamma_0^\gamma(s)-\Gamma_{-m}^\gamma(s)) \otimes L^\gamma_{-m},
  \end{aligned}
 \end{equation}
  which increments the conformal dimension. 

  In correspondence with the $T\Phi$ OPE, the $TT$ OPE maybe obtained by demanding consistency with $\mathcal V_\gamma$ commutation relations:
 \begin{equation}
  \begin{aligned}
    \left[L_m^\gamma, L_n^\gamma\right] &= \left[\oint \frac{dz[dw]}{(2\pi i)^2}\right]_Cz^{\gamma(n+1)} T(z) w^{\gamma(m+1)}T(w)\\
    &=\left[\oint \frac{dz[dw]}{(2\pi i)^2}\right]_Cz^{\gamma(n+1)} w^{\gamma(m+1)}\\&\quad\quad\quad\left(\frac{\hat c}{(z-w)^{3\gamma+1}}+\frac{(1+\gamma ) T(w)}{(z-w)^{1+\gamma}}+\frac{\partial^\gamma_w T}{z-w}\right)\\
        &=   A_{m,n}\otimes L_{m+n}^\gamma+ \omega(L^\gamma_m,L^\gamma_n)\otimes Z^\gamma
  \end{aligned}
 \end{equation}
 We schematically denote by $\hat c\in\mathcal H$ the function that characterizes the particular extension of $\mathcal W_\gamma$ we are considering. More precisely, $(\omega(L^\gamma_m,L^\gamma_n)\otimes Z^\gamma)(z^{\gamma k}) =  \delta_{m+n,0}\eta_n(k)c(k)$ where $c(k):=Z^\gamma(z^{\gamma k})$ is the degree dependent central charge.


 Collecting the two elementary OPEs,
 \begin{equation}
  \begin{aligned}
    T(z)\Phi(w) &\sim \frac{ h_\gamma \Phi}{(z-w)^{1+\gamma}}+\frac{\partial_w^\gamma \Phi}{z-w},\\
    T(z)T(w) &\sim \frac{ \hat c}{(z-w)^{3\gamma+1}}+\frac{(1+\gamma ) T_k(w)}{(z-w)^{1+\gamma}}+\frac{\partial^\gamma_w T_k}{z-w},
  \end{aligned}
 \end{equation}
  we find that the first specifies a notion of primary weight independent of $\mathcal H$-structure while the second suggests a $\mathcal H$ parameterized family of descendants where elements of the same family are $\mathcal H$-related.

  The latter is interpreted as a \emph{state dependence}. To better illustrate this, we consider the form of inner products of descendants in modules of $\mathcal V_{\gamma}$. Let 
 $$
|n_1,..,n_m \rangle:=L^\gamma_{-n_1}...L^\gamma_{-n_m}|h\rangle,
 $$   
 where $n_1,...,n_n>0$ be a generic descendant obtained from a primary of weight $h$. Then, the usual procedure of commuting through the generators to evaluate inner products of descendants motivates the inner product structure
 \begin{equation}
 \langle k_1,\cdots,k_i,\cdots,k_l   |n_1,\cdots,n_j,\cdots,n_m \rangle = \sum_{k,n} c_{k_i, n_j} (k\cdot n)_h
 \end{equation}
 where $(k\cdot n)_h$ is the part of the inner product that depends only on $h$ with $k$ and $n$ being basis vectors. The coefficients $c_{k_i, n_j}$, and hence the inner product, depend on the states that intermediate between the primary and the descendant through the evaluation of an operator valued central charge against the states. This gives rise to a conformal family structure that is much more intricate, with non-trivial mixtures of states. More explicitly, if we envision $\hat c := c+\hat{\mathcal O}$ as our operator valued central charge, and $|h,  o\rangle$ as a highest weight state labeled by eigenvalue $o$, the fractional lowering operator mixes $|h,  o\rangle$ with  $|h-n,  o\rangle$

 \section{Locality}

While the fractional Virasoro algebra can be thought of as a purely abstract generalization, it was conceived in order to make sense of nonlocal currents. Since the algebra and the fractional CFTs that realize it are built around a classically nonlocal pseudo-differential operator\begin{footnote}{We will say an operator is (clasically) nonlocal if $\supp D\phi $ is larger than $\supp \phi$}\end{footnote}, the theory is tightly intertwined with the quantum mechanical notion of locality. 

A diagnostic of the locality of a CFT is the commutation rule of its operator algebra.
 For the present context, we say that two operators, $A$ and $B$ are \emph{local with respect to one another} if it can be established that their commutator, as a formal power series is such that
 \begin{equation}\label{eq:locality}
z^{n\gamma}[A(z), B(0)] = 0.
\end{equation}
 for some positive integer $n$.
 The CFT is local if such a relationship holds for all operator pairs. This is the usual notion of locality of 2D CFTs~\cite{frenkel} extended up to coverings of Riemann surfaces for compatibility with an algebra built on $\mathbb C[[z^{\pm\gamma}]]$. Notice that in choosing this definition, we relegate any mutual nonlocality to $A$ and $B$ being related by an $\mathcal H$-function. In general, for non-trivial (non-constant) $\mathcal H$-valued fusion numbers, the locality criterion fails for any two primaries $A$ and $B$ (descendants are mutually nonlocal by construction).



\subsection{Localizing field redefinitions and a criterion for locality}

One has to be mindful of the potential existence of field redefinitions that map seemingly nonlocal dynamics to their local counterpart. In~\cite{basa_classification_2020} we argue that certain Gaussian partition functions of theories with fractional Laplacian actions are equivalent to their non-fractional counterparts under a dynamical field redefinition.  For instance,
$$
Z_\gamma = \int \mathcal D\Phi e^{\int \Phi(-\Delta)^\gamma \Phi - m^2\Phi^2}
$$
is, up to a constant, equivalent to $Z_1$ under $\Psi = (-\Delta)^{\frac{1-\gamma}{2}}\Phi$ only if $m=0$. More generally, such a \emph{localizing} field redefinition should be compatible with an integration by parts rule,
$$
\int P\Psi P\Psi,
$$
for some pseudo-differential operator $P$. The $m^2\Phi^2$ deformation, for instance, is an obstruction to finding a field redefinition compatible with any pseudo-differential integration by parts rule. See~\cite{basa_classification_2020} for more details. We corroborate the claim that the massless fractional Gaussian theory is seemingly local while its massive counterpart is not by looking at the scaling of the geometric entropy of both QFTs (Only the latter exhibits an area-law violation in the UV.) which is also consitent with~\cite{LT}.

In dimension 2 and for a CFT, we are able to explore this idea more robustly. While the fractional Virasoro algebra and its CFT realization are so far unwieldy objects, they provide insight into what quantum (non)locality really is and how it emerges. We have, then, the following Lemma that codifies these ideas.
\begin{lemma}\label{lem:1}
  The tensor product $\mathcal H _\gamma \otimes W$ inherits the structure of a multi-Lie algebra which is isomorphic to $\mathcal W_\gamma$. Furthermore, if $Vir _{c,\gamma}$ has a central charge in $\mathbb C$, then $Vir_{c,1}$ is a (central) $\mathbb C$-extension of $\mathcal H _\gamma \otimes W$.
\end{lemma}
\begin{proof}
  Let $c\in \mathbb C$ be a constant and consider the fractional central extension along with the usual integer extension tensored with $\mathcal H_\gamma$ as vector spaces:

\begin{equation*}
  \resizebox{0.65\columnwidth}{!}{%
\begin{tikzpicture}
  \node (A) {$0$};
  \node (B) [right of=A] {$\mathcal H_\gamma$};
  \node (C) [right of=B] {$\text{Vir} _{c,\gamma}$};
  \node (D) [right of=C] {$ \mathcal W _\gamma$};
  \node (E) [right of=D] {$0$};
  \node (A1) [below of = A] {$0$};
  \node (B1) [right of=A1] {$\mathcal{H}_\gamma\otimes \mathbb C $};
  \node (C1) [right of=B1] {$ \mathcal{H}_\gamma\otimes \text{Vir}_{c}$};
  \node (D1) [right of=C1] {$\mathcal{H}_\gamma\otimes W$};
  \node (E1) [right of=D1] {$0$};
  \draw[->] (A) to node {} (B);
  \draw[->] (B) to node {} (C);
  \draw[->] (C) to node {} (D);
  \draw[->] (D) to node {} (E);
  \draw[->] (A1) to node {} (B1);
  \draw[->] (B1) to node {} (C1);
  \draw[->] (C1) to node {} (D1);
  \draw[->] (D1) to node {} (E1);
  \draw[<->, dashed] (B1) to node {} (B);
  \draw[<->, dashed] (C1) to node {} (C);
  \draw[<->, dashed] (D1) to node {} (D);
\end{tikzpicture}%
}
\end{equation*}
Clearly, $\mathcal{H}_\gamma\otimes \mathbb C \cong \mathcal H_\gamma$ up to a constant. The Lie bracket on $\mathcal H _\gamma \otimes W$, with $W$ the integer Witt algebra, is given by
\begin{equation}
[\phi\otimes L_n, \psi \otimes L_m]= [\phi, \psi]_{m,n} \otimes L_{m+n}.
\end{equation}
Thus, $\mathcal H_\gamma \otimes W \cong \mathcal W_\gamma$. Then, we have the middle isomorphism,
\beq
\text{Vir}_{c, \gamma}\cong \mathcal{H}_\gamma\otimes \text{Vir}_{c,1}.
\eeq
\end{proof}

Physically, the fractional theories that are local by our definition (which allows for coverings) are those with constant central charge because the object $\text{Vir}_{c,1}\otimes \mathcal H_\gamma$ is generated by a branched stress tensor with fixed central charge (which is local with respect to itself). Also, note that the existence of the Witt algebra isomorphism can be understood as the existence of a field redefinition that translates to the existence of a `localizing' redefinition at the level of OPEs. Then, we are in a position to verify the conjecture of Ref.~\cite{basa_classification_2020} in the context of 2D CFT, namely that a theory is nonlocal if there does not exist a localizing field redefinition.

\section{Comments on nonlocal field theories}\label{sec:QFT}


\subsection{Fractional bosonic CFT}\label{sec:boson}

The simplest actionable fractional 2D CFT one can consider is
\begin{equation}
  \label{eq:frac_boson}
  S = g\int d^2 z \partial^\gamma\Phi \bar \partial^\gamma \Phi.
\end{equation}
 The conformal invariance of the fractional bosonic CFT has been well established~\cite{paulos_conformal_2016,rajabpour_conformal_2011} as we mentioned previously and can be shown to follow from the CS extension theorem~\cite{caffarelli_extension_2007}
which is of course equivalent to the well studied fractional free Gaussian theory. This theory, while built out of a nonlocal kinetic operator, does not furnish a generalized operator product algebra with an operator-valued central charge. While this fact is not surprising in light of Lemma 1, it is instructive to see how it comes about. To this end, we will exhibit a field redefinition that localizes the action in an approach that resembles techniques such as bosonization/fermionization.


Consider the field redefinition $\Phi = \partial^{1-\gamma} \Phi'$ under which one has, up to the singular terms of the binomial expansion, 
\begin{equation}
\begin{aligned}
  \partial &^\gamma \partial^{1-\gamma} \Phi'(z) \\ &\sim  \int \frac{d\xi d\eta}{(2\pi i)^2} \frac{\Phi' (\eta)}{(\eta-\xi)^{2-\gamma}}\frac{1}{(\eta -z)^{1+\gamma}}\\
  &\sim\sum_{kl} \int \frac{d\xi d\eta}{(2\pi i)^2}\Phi'(\eta) \eta^{-3-k-l}(-\xi)^{k}(-z)^{l}\\
  &\sim\sum_{kl} \int \frac{d\xi d\eta}{(2\pi i)^2}\Phi'(\eta) \eta^{-3-k-l}(-\xi)^{k}(-z)^{l}\\
  &\sim \int \frac{d\xi}{2\pi i}\frac{\Phi'(\xi)}{(z-\xi)^2}\\
  &\sim \partial \Phi'(z),
\end{aligned}
\end{equation}
where $\sim$ denotes an equivalence of analytical structure. Of course, arbitrary field redefinitions cannot be used.  We say that a field redefinition is physically meaningful if the induced transformation of the partition function is simply a rescaling. That such a \emph{non-trivial} field redefinition exists is an unusual property of QFTs in general.

Under the aforementioned field redefinition, one has
\begin{equation}
  \frac{1}{(z-w)^{n+\gamma}}\mapsto \frac{1}{(z-w)^{n+1}}
\end{equation}
and hence the OPE for the fractional free Boson can readily be constructed:
\begin{equation}
\label{eq:phiphi}
\begin{aligned}
    \partial^\gamma \Phi(z)\partial^\gamma \Phi(w) &\sim \frac{1}{(z-w)^{1+\gamma}}\\
    T(z)\partial ^\gamma \Phi(w)&\sim\frac{(2-\gamma)\partial^\gamma_w \Phi(w)}{(z-w)^{1+\gamma}}+\frac{\partial^{2,\gamma}_w\Phi(w)}{z-w}\\
    T(z)T(w)&\sim \frac{c_\gamma/2}{(z-w)^{1+3\gamma}}+\text{Laurent exp.}
  \end{aligned}
\end{equation}

Upon removing the branches of the singular terms by taking a cover of the fractional OPE, one obtains the usual free scalar OPE with $c_{\gamma\to1}=1$, which by our definition makes this theory local.

\subsection{Deformations and Renormalization}\label{sec:defo}

With the understanding that the free fractional Boson is a fractional CFT with constant central charge, we now motivate the conceptual framework behind a perturbation theory built around such a fixed point to better understand how nonlocality emerges under deformations and, in contrast to the preceding example, how one can identify fractional CFTs with nonconstant central charge.


Knowledge of the scaling dimensions and the coefficients of the 3-point function are sufficient to generate the (one loop) beta functions for the various relevant couplings of the perturbative QFT. For coupling $g_i$ associated with a deformation of dimension $h_i$, the $\beta$-function reads
\begin{equation}\label{eq:beta}
    \beta(g_i) = (d-h_i)g_i- \mathcal N^{jk}_i g_jg_k  +\cdots.
\end{equation}
This expression encodes the effect of a deformation of the Gaussian fixed point by the relevant local operators at the level of coupling constants. The field redefinition arguments of~\cite{basa_classification_2020} $\beta(g_i)>0$ imply the non existence of a localizing field redefinition. We can argue, as a first step, that
 the perturbative deformations of a 2D fractional CFT stress tensor via a set of operators that are mutually local with respect to one another generate a flow to a nonlocal fractional CFT with $c\in\mathcal H$ in general.

In the simplest incarnation of this idea, if $\{:\phi^n:\}_{n\in [1, m]}$ is the set of local deformations, then $\phi$ is mutually local with respect to itself and its integer powers but not necessarily the stress tensor of the unperturbed fractional CFT. Given a perturbative polynomial deformation, one has for the stress tensor at an accessible  conformal fixed point
$$
T=T_0^{(\gamma)}+\delta T^{(1)},
$$
where the superscript denotes the value of $\gamma$ for which Eq.~\eqref{eq:locality} holds for $\phi$ paired either with (powers of) itself or the free stress tensor. The idea here is to realize that because $T_0^{(\gamma)}$ is mutually nonlocal wrt $\delta T^{(1)}$ in general, $T_0^{(\gamma)}\delta T^{(1)}$  differs from $T_0^{(\gamma)}T_0^{(\gamma)}$ or $TT$ by a strictly non-constant $\mathcal H$ function at leading order.



\subsection{A C-theorem for fractional QFT}
We had established that constant central charge theories constructed out of nonlocal operators have equivalent local representations. 
The intuition of the sketch in Sec.~\ref{sec:defo} implies that this equivalence is not stable to perturbations, however. 
Suppose $\mathcal O_i := \partial_{g_i} \mathcal L (g, \Lambda_{UV}, \mu)$ for $\mathcal L$ some Lagrangian and define 
\begin{equation} \beta_i:=\frac{dg_i}{dlog\mu}\end{equation}
Renormalizabily is equivalent to the condition that 
\beq 
\text{tr } T = \beta ^i (g) \mathcal O_i
\eeq
where $\beta ^i (g)$ are the components of the vector field that generate the renormalization group. If $T$ is not a {\it local} stress energy tensor, we also have to note the fusion coeffcients are valued in $\mathcal H$ and hence are degree/state dependent:
\begin{equation}
  \phi\otimes\mathcal O_i\times \psi \otimes \mathcal O_j = \sum_k (\phi \star_{i,j}\psi)\otimes \mathcal O_k.
\end{equation}
Assume that there exists subspaces of the fractional graded vector space of states, $V=\bigoplus_{k\in\mathbb Z} V^{\gamma k}$, along the flow and at the target conformal fixed point. For each $\gamma$, the grading gives rise to a sequence,
\begin{equation}
   G_{ij}(k) := \eta^{3\gamma +1} \langle \mathcal O_i(\eta)\mathcal O_j(0) \rangle\big\vert_{V^{\gamma k}},
  \end{equation}
  on the configuration space of the theory.
We also define the following quantities, following Zamolodchikov in \cite{zamo}
\begin{equation}
C(k)= 2 \eta^{3\gamma +1} \langle T(\eta) T(0)\rangle\big\vert_{V^{\gamma k}}
\end{equation}
and
\begin{equation}
H_i(k)= 2 \eta^{3\gamma +1} \langle T(\eta) \mathcal O_i(0)\rangle\big\vert_{V^{\gamma k}}
\end{equation}
 in order to formulate a graded version of the C-theorem~\cite{zamo}:
\begin{theorem}
  If $V^{\gamma k}$ contains no negative norm states, $G_{ij}(k)$ is a sequence of metrics on $\mathcal G$. Then, the C-theorem holds independently for each $V^{\gamma k}$:
  \begin{equation}
  c(g;k)=C(k) + 4 \beta ^iH_i(k)- 6 \beta ^i \beta ^j G_{ij}(k)
  \end{equation} 
 satisfies the following properties
 \begin{enumerate}
 \item $c$ is non-increasing along the graded flow
  \begin{equation}\dot c\vert_k \leq 0\end{equation}
 \item critical points (i.e. points for which $\beta^i(g^*)=0$) are stationary points for $c(g;k)$
 \begin{equation}\label{eq:crit}
 \frac{\partial c(g;k)}{\partial g^i}\bigg\vert_{g^*}=0
 \end{equation}
 
 \item 
At a critical $g^*$ the theory has the symmetry of the fractional Virasoro algebra with central charge equal to the value of $c(g;k)$ at $g=g^*$.
 \end{enumerate}
\end{theorem}
\begin{proof}
  The proof follows the one in~\cite{zamo} \emph{mutatis mutandis}. With the insight of Sec.~\ref{sec:geo}, the C-theorem holds independently on restrictions to fibers over $D$. 
\end{proof}
As an immediate corollary, we can formalize the example in Sec.~\ref{sec:defo}:

\begin{corollary}\label{cor-ctheor}
 An RG flow fixed  point $g^*\in\mathcal G$ in dimension 2 is a local CFT if and only if $c(g^*)$ is not state-dependent ($k$-invariant).
 
 \end{corollary}
We therefore have a notion of a graded RG flow compatible with both the trivial $\mathcal H$ structure of the Gaussian fixed point and the non-trivial $\mathcal H$ structure of the target nonlocal fractional fixed point.
%

\subsection{Renormalizibility and the LRI model}\label{sec:LRI}

 It is rather fortuitous that there exists a conformal fixed point that is a candidate for being a non-trivial example of a nonlocal CFT that has been studied extensively. The long range Ising model is an extension of the standard Ising model where the spin correlations extend over the entire lattice, $\Lambda$, and decay as a fractional power law
 \beq
H=-J\sum_{i,j\in \Lambda } \sigma_i\sigma_j/r_{ij}^{d+2\gamma}, \hspace{2mm J>0}.
\eeq
We fix $d=2$ to remain consistent with our model. This particular all-to-all interaction of spins, in the continuum limit, is associated with the fractional Laplacian. Just as the usual $\phi^4$ theory belongs to the same universality class as the short-range Ising model, the nonlocal $\phi^4$ theory,
\begin{equation}\label{eq:LRI}
  S = \int_{\mathbb{R}^2}\phi (-\Delta)^\gamma \phi +g_4\phi^4,
\end{equation}
can encode LRI physics in the IR. If the IR fixed point of this theory is conformally invariant, Corollary~\ref{cor-ctheor} implies that it is an example of a generalized CFT that is nonlocal by our definition. Ref.~\cite{paulos_conformal_2016} demonstrates that LRI fixed point is indeed conformally invariant, albeit without a local stress tensor. Our construction describes the fixed point CFT as one possessing a parameterized family of conformal symmetries captured by a nonlocal stress tensor. There are numerous directions to take a discussion regarding a LRI CFT, especially through the lens of state-dependence arising from nonlocality. We focus on the implication of our theory for the phase diagram.

The foundational work of Fisher et. al. and Sak~\cite{fisher_critical_1972,sak_recursion_1973,sak_1977} has become the standard theory of the LRI model. Counting powers in Eq.~\eqref{eq:LRI}, one finds that for $\gamma>1/2$, the theory flows to the LRI fixed point under the relevant $\phi^4$ interaction. Thus, at $\gamma_{*}:=1/2$ one expects a transition from trivial fractional Gaussian physics to LRI physics. Note, by Lemma~\ref{lem:1}, that this is a local to nonlocal transition. We reiterate here that this claim of locality can be independently tested by computing the UV asymptotics of the entanglement entropy of a free fractional Gaussian fixed point to find that there is no power law dependence on $\gamma$. In other words, the area law holds. Tuning $\gamma$ beyond $\gamma_*$ is believed to eventually drive a \emph{continuous} localization transition where SRI physics is obtained above $\gamma_{**}\leq 1$. The reasoning, again, is based on dimension counting: There exists a $\gamma_{**}$ such that the dimension of $\phi$ of the long range theory, which is fixed by the kinetic term and not anomalous, decreases to match what it would have been in a short range theory.

If primary operator spectra are to be continuously related, one expects the descendants and hence Virasoro algebras to map continuously as well. As we have discussed, the simplest fractional Virasoro algebra one can envision is of the form $\mathcal H\otimes \text{Vir}_c$ (which is not representative of the LRI CFT). Only in this case can one take a branched covering and relate the two Virasoro algebras continuously. Thus, the continuity of the transition must be weakened to the level of correlators, not operator algebras~\cite{Honkonen_1989,behan_long-range_2017,numerics1,blanchard_influence_2013}. While the nature of this transition has long been contested, the objection against the continuity of the transition that raises the contradiction that the LRI spectrum contains two parity odd primaries, $\phi$ and $\phi^3$, while the SRI hosts only a single relevant parity odd primary, $\phi$, from which $\phi^3$ descends appears to be consistent with our heuristic claim. We refer to~\cite{behan_long-range_2017,behan_scaling_2017} for further elaboration on this problem and their proposed resolution. 

From our perspective, the former transition can be understood from simple perturbation theory: The fractional Gaussian fixed point has the physical content of a local theory with the operator algebra expressed in a peculiar basis. This is not the case at $\gamma=\gamma_{**}$, where there is presumed to be a transition from LRI to SRI physics. We would expect such a crossover to be state-dependent with the operator spectra being related by non-constant $\mathcal H$ functions. Assuming the LRI fixed point is characterized by a nonlocal Virasoro algebra, this transition cannot be continuous in the traditional OPE sense and \emph{will} lead to contradictions under standard theoretical prescriptions. 


\section{Final Remarks}
The construction of a CFT arising from the fractional Virasoro algebra points one in the unfamiliar direction of a generally non-constant central charge, one that depends on the state.  In particular, the notion of locality is refined to allow for partition functions preserving field redefinitions, placing our work relying on entanglement entropy scaling arguments in~\cite{basa_classification_2020} on firm mathematical footing. Furthermore, while so far forbiddingly unwieldy to work with explicitly, we establish that perturbative RG techniques, guided by a degree dependent refinement of the C-theorem, applied to nonlocal field theories can yield examples of such nonlocal fractional CFTs. The landscape of theories changes remarkably when conformal fixed points are allowed to be state-dependent. This state-dependence is assumed to be controlled by higher dimensional moduli in a manner consistent with the CS extension theorem.

Our work here motivates a shift in thinking from specific nonlocal theories to the classifying spaces of such theories that, as a scheme of conformal perturbation theory, encodes correctly the parameterized family of conformal symmetries of the UV and IR CFTs implied by a nonlocal Virasoro algebra. There is a rather striking analogy to be made here.  Namely,  physics of fundamental degrees of freedom is described by the embedding of string world sheets in some geometry UV completes a local QFT by giving it extended, non-singular structure that survives at short distances. Curiously, the corresponding UV dynamics is best understood not through any particular worldsheet but rather the moduli space of complex structures which, geometrically, is rather similar in spirit to Sec.~\ref{sec:geo}. It is perhaps meaningful to interpret a nonlocal QFT arising from the fractional Laplacian as an effective UV completion given by extending world lines not to sheets but to $1+\gamma$ dimensional objects which not only provides justification to invoking rather sophisticated mathematics to describe theories that might initially seem within the grasp of conventional methods of C/QFT but also positions us to make contact with the fractality of quantum chaotic processes and strongly correlated quantum criticality~\cite{fractalfs}. In the former case, the connection with quantum chaos should build on generalizations of the Cantor sets in the context of von Neumann algebras. It may be possible to explore these avenues by constructing sigma models where the base space is a metric measure space in which the worldline embeds as a $1+\gamma$ dimensional object.

 Going beyond this analogy, it would be interesting to interpret our fractional CFTs as those on a (fractional) string worldsheet and perhaps make physical sense of nonlocal BRST quantization. In this context, the connection between the nonlocal Virasoro algebra and the spaces that realize their implied symmetry, which we expect to be a stacky geometric structure, may be related to the question of symmetries in bulk quantum gravity~\cite{harlow_symmetries_2019}. Building on the corresponding geometric interpretation and its implications for the standard string worldsheet and generalizations of these ideas to other vertex algebras will be the focus of future work.


%



\paragraph{Funding information}
P.W.P. thanks DMR19-19143 for partial funding of this project.




\bibliographystyle{SciPost_bibstyle} 

\end{document}